\documentclass[conference]{IEEEtran}                                      %Double column single space version default
\usepackage{rotating}
\usepackage{amsmath}
\ifCLASSINFOpdf
\usepackage{graphicx,amssymb}
\usepackage{algorithm}
 \usepackage{algorithmic}
 \usepackage{diagbox}
\usepackage{multirow}
\usepackage{subcaption}
\usepackage{makecell}

\makeatletter

\@addtoreset{algorithm}{section}
\makeatother
\usepackage{pdfsync}
\usepackage{amsthm}

\newtheorem{theorem}{Theorem}[section]

\renewcommand{\theequation} {\thesection.\arabic{equation}}

\numberwithin{equation}{section}
\usepackage{color}
\begin{document}

\title{   Iterative polynomial approximation algorithms for inverse graph filters} % spatially distributed systems}
\author{Cheng Cheng, Qiyu Sun, Cong Zheng} %(Member, IEEE)
\thanks{Cheng is with the School of Mathematics, Sun Yat-sen University, Guangzhou, 510275;
\and Sun and Zheng are with the  Department of Mathematics,
University of Central Florida,
Orlando, Florida 32816.
Emails:  chengch66@mail.sysu.edu.cn;  Cheng is with the Department of Mathematics qiyu.sun@ucf.edu. }

\maketitle

\begin{abstract}
 Chebyshev interpolation polynomials exhibit  the exponential approximation property to analytic functions on a cube. Based on
the Chebyshev interpolation polynomial approximation, we propose
 iterative polynomial approximation algorithms to implement the inverse filter with a polynomial graph filter of  commutative graph shifts in a distributed manner. The proposed algorithms exhibit exponential convergence properties, and they can be implemented on distributed networks in which agents are equipped with a data processing
subsystem  for limited data storage and computation power,
and with  a  one-hop communication subsystem   for  direct data exchange only with their adjacent agents.
Our simulations show that the proposed polynomial approximation algorithms may
converge faster
than the Chebyshev polynomial approximation algorithm
 and the conventional gradient descent algorithm 
 do. 
\end{abstract}


\vskip-1.5mm

\section{\bf Introduction}  %\noindent
 Graphs  are widely used to model the complicated   topological structure of networks, such as (wireless) sensor networks, smart
grids and social networks \cite{bronstein17, dong20, ncjs22, Ortega18, sandryhaila14,  sandryhaila13, shuman13,     stankovic2019introduction}. Many data sets on a network can be represented by signals  ${\bf x}=(x_i)_{i\in V}$ residing on the graph ${\mathcal G} = (V, E)$, % ${\mathcal G}=(V, E)$,
 where $x_i$ represents the real/complex/vector-valued data at the vertex/agent $i\in V$, a vertex in $V$ may represent an agent of the network, and an edge in $E$ between vertices could indicate that the
corresponding agents  have a
 peer-to-peer communication link between
them.  Graph signal processing paves an innovative way to extract valuable insights from datasets residing on the complicated networks,  % \cite{Wass94}-\cite{Cheng19}.
\cite{Cheng19, chong2003, %Hebner17,
  mfa07, Motee17,   %Wass94,
    Yick08}.

 The filtering procedure for signals on a network is a linear transformation
%\st{Similarly, a linear operation on data sets on a network can be  formulated as a linear filtering procedure}
   \vspace{-.6em} \begin{equation}\label{filtering.def}
{\bf x}\longmapsto {\bf y}={\bf H}{\bf x},    \vspace{-.4em}\end{equation}
which maps a graph signal   ${\bf x}$
 to another graph signal ${\bf y}={\bf H}{\bf x}$, and ${\bf H}=(H(i,j))_{ i,j\in V}$ is known as a {\em graph filter}.
 An {\em graph shift}  is an elementary graph filter, and we say that a matrix ${\bf S}=(S(i,j))_{i,j\in V}$  on  the  graph ${\mathcal G}=(V, E)$ if
$S(i,j)\ne 0$ only if either $j=i$ or $(i,j)\in E$. 
%Illustrative examples are the adjacency matrix ${\bf A}$, Laplacian matrix ${\bf L}={\bf D}-{\bf A}$, and
 %symmetrically normalized  Laplacian ${\bf L}^{\rm sym}:={\bf D}^{-1/2} {\bf L}{\bf D}^{-1/2}$, where ${\bf D}$
 %is the degree matrix of the graph
%  \cite{sandryhaila13}, \cite{ncjs22}-\cite{jiang19}.
%\cite{ncjs22, Coutino17, jiang19, sandryhaila13,     segarra17}.
In \cite{ncjs22},   the notion of multiple commutative graph shifts ${\bf S}_1, \dots,  {\bf S}_d$ are introduced,
%	\vspace{-0.4em}
\begin{equation}\label{commutativityS}
	{\bf S}_k{\bf S}_{k'}={\bf S}_{k'}{\bf S}_k,\  1\le k,k'\le d,
%	\vspace{-0.6em}
\end{equation}
and some multiple commutative graph shifts on circulant/Cayley graphs and on Cartesian product graphs  are constructed with physical interpretation.  An important property for   commutative graph shifts ${\mathbf S}_1, \ldots, {\mathbf S}_d$ is that they
can be upper-triangularized  simultaneously,
\begin{equation}
	\label{upperdiagonalization}
\widehat{\bf S}_k={\bf U}^{\rm H}{\bf S}_k{\bf U},\  1\le k\le d,
	\end{equation}
where ${\bf U}$ is a unitary matrix  % ${\bf U}^{\rm H}$ is the Hermitian of the matrix ${\bf U}$,
 and
	$\widehat{\bf S}_k=(\widehat S_{k}(i,j))_{1\le i, j\le N}, 1\le k\le d$, are upper triangular matrices.
% \cite[Theorem 2.3.3]{horn1990matrix}.
As
 $\widehat{S}_k(i, i), 1\le i\le N$, are eigenvalues of ${\bf S}_k, 1\le k\le d$,
 we call the set	
 \begin{equation}\label{jointspectrum.def} \Lambda=\big\{\pmb \lambda_i=\big(\widehat{S}_1(i,i), ..., \widehat{ S}_d(i,i)\big), 1\le i\le N\big\}\end{equation}
as the {\em joint spectrum} of  ${\bf S}_1, \ldots, {\bf S}_d$  \cite{ncjs22}.
  For the case that graph shifts ${\bf S}_1, \ldots, {\bf S}_d$ are symmetric, all $\widehat S_k(i,i), 1\le i\le N, 1\le k\le d$
  are real-valued and  the joint spectrum of %commutative
  graph shifts ${\bf S}_1, \ldots, {\bf S}_d$ is contained in some cube,
%$[{\pmb \mu}, {\pmb \nu}]=[\mu_1,\nu_1] \times \cdots \times[\mu_d,\nu_d]$   in ${\mathbb R}^d$, i.e.,
\begin{equation}\label{jointspectralcubic.def}
\Lambda\subset [{\pmb \mu}, {\pmb \nu}]:=[\mu_1,\nu_1] \times \cdots \times[\mu_d,\nu_d]\subset {\mathbb R}^d,
\end{equation}
 where  for each $1\le k\le d$,  $[\mu_k, \nu_k]$ is the (minimal) interval to contain the spectrum of the graph shift ${\bf S}_k$.
 
 A popular family of graph filters is
    {\em polynomial  filters}  of  commutative graph shifts   ${\bf S}_1, \dots,  {\bf S}_d$,
    \vspace{-0.6em}
\begin{equation}\label{MultiShiftPolynomial}
	{\bf H}=h({\bf S}_1, \ldots, {\bf S}_d)=\sum_{ l_1=0}^{L_1} \cdots \sum_{ l_d=0}^{L_d}  h_{l_1,\dots,l_d}{\bf S}_1^{l_1}\cdots {\bf S}_d^{l_d},
	\vspace{-0.6em}
\end{equation}
 %  as its blocks,
   where
$h$ is a multivariate polynomial
in variables $t_1,\cdots,t_d$,
\vspace{-0.6em}
% \begin{equation*}\label{MultiShiftPolynomial.polynomial}
$$h(t_1, \ldots, t_d)=\sum_{ l_1=0}^{L_1} \cdots \sum_{ l_d=0}^{L_d}  h_{l_1,\dots,l_d} t_1^{l_1} \ldots t_d^{l_d}$$
 % of degree $\sum_{k=1}^d L_k$.
%h_{l_1,\dots,l_d},  0\le l_k\le L_k, 1\le k\le d$ are the polynomial coefficients for the terms  ${\bf S}_1^{l_1}\cdots {\bf S}_d^{l_d}$.
  %\cite{ncjs22, segarra17},  \cite{Leus17}-\cite{David2019}.
    \cite{Emirov19,  mario19, ncjs22, Leus17, David2019,  Lu18, segarra17, Shuman18,  Waheed18}.  The  consideration of  polynomial  filters of multiple graph
shifts  is mainly motivated by signal processing of time-varying data sets on a sensor network over a period of time, which carry
different correlation characteristics for spatial-temporal directions. 
%  For  polynomial   filters  in \eqref{MultiShiftPolynomial},  a significant advantage is that
%the corresponding filtering procedure \eqref{filtering.def}
%   can be implemented at the vertex level in which  each vertex  is equipped
%%with a data processing subsystem for limited data storage and computation power,
%%and
%with a {\bf one-hop  communication subsystem}, i.e.,  each agent has direct data exchange only with its adjacent  agents, see
%   \cite[Algorithms 1 and 2]{ncjs22} and  also \cite{mario19} for the  distributed one-hop implementation of more general edge-variant FIR graph filters. 
 
  Inverse filtering procedure associated with  a polynomial filter of graph shifts
is a versatile tool, offering a wide applications across denoising,  signal reconstruction, graph semi-supervised
learning and many other applications
%\cite{jiang19, Leus17, Lu18}-\cite{Emirov19}, \cite{siheng_inpaint15}-\cite{cheng2021}.
\cite{siheng_inpaint15, sihengTV15, cheng2021,  Emirov19, mario19,   Leus17, jiang19,  Lu18,         Onuki16, Shi15, Shuman18}. Its importance lies in its ability to recover original signals from observed data, enabling a deeper understanding of underlying network structures and dynamics. In this paper, we consider distributed implementation of inverse filtering procedure on  simple graphs   (i.e.,
 unweighted undirected graphs containing no  loops or multiple edges)
   of large order $N\ge 1$.

Given a polynomial filter ${\bf H}$ of graph shifts, one
 of the main challenges in the corresponding inverse filtering procedure
\begin{equation}\label{inverseprocedure}
{\bf y}\longmapsto {\bf x}={\bf H}^{-1}{\bf y}
\end{equation}
 is on its distributed implementation, as
  the inverse filter  ${\bf H}^{-1}$ is usually not a polynomial filter of small degree even if ${\bf H}$ is.
%  The inverse filtering procedure can be considered as solving the linear system
%  ${\bf H}{\bf x}={\bf y}$.
The first two authors of this paper proposed the following exponentially convergent quasi-Newton method
with arbitrary initial ${\bf x}^{(0)}$,
\vspace{-.6em} \begin{equation}\label{Approximationalgorithm}
 {\bf e}^{(m)}= {\bf H}{\bf x}^{(m-1)}-{\bf y}\ \ {\rm and}\  \
{\bf x}^{(m)}={\bf x}^{(m-1)}-{\bf G}{\bf e}^{(m)}, \ m\ge 1
\end{equation}
  to fulfill the inverse filtering procedure,
% where
%%the approximation filter ${\bf G}$ to ${\bf H}^{-1}$  is  of a polynomial form,
%the approximation filter ${\bf G}$ is so chosen that  the spectral radius of ${\bf I}-{\bf G}{\bf H}$  is strictly less than $1$
% \cite{ncjs22, Emirov19, cheng2021}.
 where the polynomial approximation filter ${\bf G}$  to the inverse ${\bf H}^{-1}$
  is  so chosen that
  the spectral radius of ${\bf I}-{\bf G}{\bf H}$  is strictly less than $1$
 \cite{cheng2021, Emirov19,  ncjs22}. 
%    More importantly,  each iteration in  \eqref{Approximationalgorithm}
% includes mainly two  filtering procedures associated with  polynomial filters ${\bf H}$ and ${\bf G}$.  
An important  problem  is how to select the polynomial  approximation filter ${\bf G}$ appropriately for  the fast convergence of the
 quasi-Newton method \eqref{Approximationalgorithm}.
The above  problem  has been well studied when
${\bf H}$ is  a polynomial filter  of the  % (symmetrically normalized)
 graph Laplacian (and a single graph shift in general)
 \cite{%sihengTV15,
 sihengTV15,  Emirov19, isufi19,
 Leus17, Shi15,  Shuman18}.
% , even though the original polynomial filter ${\bf H}$ and the polynomial approximation filter ${\bf G}$
% may not satisfy the geodesic-width requirement \eqref{width.req},
For a polynomial filter ${\bf H}$
 of multiple  graph shifts,
 optimal/Chebyshev polynomial  approximation filters
are introduced in \cite{ncjs22}. 
  In Section \ref{preliminaries.section},  
we show the  exponential approximation property of  Chebyshev interpolation polynomials to analytic functions on a cube,  and then 
 we introduce the Chebyshev interpolation polynomial filter  to approximate the inverse filter ${\bf H}^{-1}$  and the  corresponding quasi-Newton method algorithm  \eqref{cipa.def}  to implement the inverse filtering procedure \eqref
 {inverseprocedure}. In Section \ref{cipa.sec}, we show the corresponding Chebyshev interpolation polynomial  algorithm is of exponential convergence and 
  can be applied  to implement
the  inverse filtering procedure \eqref{inverseprocedure} associated with a polynomial filter   on distributed networks, see  Theorem \ref{exponentialconvergence.cipa.thm} and 
Algorithms \ref{CIPA.algorithm}. 
  Numerical experiments in Section \ref{num.sec} indicate that the proposed Chebyshev interpolation polynomial approximation algorithm have  better performance  than  Chebyshev  polynomial approximation algorithm (CPA),
   the gradient descent algorithm with optimal step size (OGDA) and  the autoregressive moving-average model (ARMA) 
  do   \cite{sihengTV15, ncjs22, isufi19, Leus17, jiang19, Shi15, Shuman18, Waheed18}. 

{\bf Notation:} \  Bold lower cases and capitals are used to represent the column vectors and matrices respectively. Define    $\|{\bf x}\|_2=(\sum_{i\in V} |x_i|^2)^{1/2}$ and  $\|{\bf x}\|_\infty=\sup_{i\in V} \{|x_i|\}$ for a graph signal ${\bf x}=(x_i)_{i\in V}$
and  $\|{\bf A}\|_2=\sup_{\|{\bf x}\|_2=1} \|{\bf A}{\bf x}\|_2$ for a graph filter ${\bf A}$.
Denote  the transpose, the Hermitian and Frobenius norm of a matrix ${\bf A}$ by
${\bf A}^T$, ${\bf A}^H$  and $\|{\bf A}\|_{\rm F}$ respectively.

\section{Chebyshev interpolating polynomials} % on commutative graph shifts and Jacobi polynomials and  Chebyshev interpolating polynomials}
\label{preliminaries.section}

%In this section,  we  define   Chebyshev interpolating polynomials on  Chebyshev nodes, and show that  they provide exponential approximations  to the reciprocal of the polynomial $h$
% on the cube
%$[{\pmb \mu}, {\pmb \nu}]$
% %\cite{Ismail2009}-\cite{xiang2012}.
% \cite{Ismail2009, Shen2011, trefethen2013, wang2011, xiang2012}.
%Our numerical simulations indicate that interpolation polynomials
%at Chebyshev points   provide better approximations to the reciprocal  of a  
%polynomial  on a cube
%   than  Chebyshev polynomials  do
%    \cite{ncjs22}.  %, see    Figure \ref{approximation.fig} and Table \ref{MaxAppErr.tab}.

Let $[{\pmb \mu}, {\pmb \nu}]=[\mu_1,\nu_1] \times \cdots \times[\mu_d,\nu_d]$ be a cube in ${\mathbb R}^d$ and 
 $h$ be a multivariate polynomial that does not vanish on the cube
 $[{\pmb \mu}, {\pmb \nu}]$, i.e.,
 \begin{equation}\label{polynomial.assump}
    h({\bf t})\neq 0\  \text{ for all }\  {\bf t}\in [{\pmb \mu}, {\pmb \nu}].
\end{equation}
Write 
 ${\bf t}_{{\bf j}; {\pmb \mu}, {\pmb \nu} }=
 ({t}_{j_1; \mu_1, \nu_1},  \ldots, {t}_{j_d; \mu_d, \nu_d}), 0\le j_k\le M, 1\le k\le d$,
  be   rescaled Chebyshev points in the cube $[{\pmb \mu}, {\pmb \nu}]$,
  and
 the Lagrange basis at rescaled Chebyshev points be defined by
$$\ell_{M}({\bf t},  {\bf t}_{{\bf j}; {\pmb \mu}, {\pmb \nu} }  )=\prod_{k=1}^d
\prod_{0\le i_k\le M, i_k\ne j_k} \frac{t_k-{ t}_{i_k; \mu_k, \nu_k}}
 {{ t}_{j_k; \mu_k, \nu_k}-{ t}_{i_k; \mu_k, \nu_k}}, $$
 where  %for $0\le j_k\le M, 1\le k\le d$, rescaled Chebyshev points on $[\mu_k, \nu_k]$ are given by
% $$ \ell_{M, \mu_k, \nu_k}(t_k, { t}_{j_k; \mu_k, \nu_k})=\prod_{1\le i\le M+1, i\ne j_k} \frac{t_k-{ t}_{i; \mu_k, \nu_k}}
% {{ t}_{j_k; \mu_k, \nu_k}-{ t}_{i; \mu_k, \nu_k}}
% $$
 $${ t}_{j_k; \mu_k, \nu_k}= \frac{\nu_k+\mu_k}{2}+ \frac{\nu_k-\mu_k}{2}
  \cos \frac{(j_k+1/2)\pi}{M+1}. $$
%  \ell_{M, \mu_k, \nu_k}(t_k, { t}_{j_k; \mu_k, \nu_k}), \ {\bf t}=(t_1, \ldots, t_d), $$
% {\textstyle \leq k}{\textstyle \leq k} for those nodes is the set of polynomials {\textstyle \{\ell _{0}(x),\ell _{1}(x),\ldots ,\ell _{k}(x)\}}{\textstyle \{\ell _{0}(x),\ell _{1}(x),\ldots ,\ell _{k}(x)\}} each of degree {\textstyle k}{\textstyle k} which take values {\textstyle \ell _{j}(x_{m})=0}{\textstyle \ell _{j}(x_{m})=0} if {\textstyle m\neq j}{\textstyle m\neq j} and {\textstyle \ell _{j}(x_{j})=1}{\textstyle \ell _{j}(x_{j})=1}
%
For a polynomial $h$  satisfying
 \eqref{polynomial.assump}, an excellent method of approximating the reciprocal   $1/h$ on the cube $[{\pmb \mu}, {\bf \pmb \nu}]$
is the Chebyshev   interpolation polynomial
 \begin{equation}\label{chebyshevinterpolation.def}
 C_M({\bf t})=\sum_{\|{\bf j}\|_\infty \le M} \frac{1}{h( {\bf t}_{{\bf j}; {\pmb \mu}, {\pmb \nu} })}
 \ell_{M}({\bf t},  {\bf t}_{{\bf j}; {\pmb \mu}, {\pmb \nu} }  ),
 \end{equation}
 which is the unique polynomial of the form
 $\sum_{\|{\bf n}\|_\infty\le M} d_{\bf n}{\bf t}^{\bf n}$ for some ${\bf d}_{\bf n}, \|{\bf n}\|_\infty\le M$, satisfying the interpolation property
  \begin{equation}
\label{chebyshevinterpolation.def1}
C_M({\bf t}_{{\bf j}; {\pmb \mu}, {\pmb \nu} })= \frac{1}{h( {\bf t}_{{\bf j}; {\pmb \mu}, {\pmb \nu} })} \ {\rm for \ all} \ \|{\bf j}\|_\infty\le M.
\end{equation}
%$C_M({\bf t}_{{\bf j}; {\pmb \mu}, {\pmb \nu} })= 1/h( {\bf t}_{{\bf j}; {\pmb \mu}, {\pmb \nu} })$,
% {\color{red} add mention how to get $d_{\bf n}$ and add a reference} where
% , \ldots,
%  \frac{\nu_d-\mu_d}{2} \cos \frac{(j_d-1/2)\pi}{N+1}\right),$$
% where
Recall that the  Lebesgue constant for the above polynomial interpolation at rescaled Chebyshev points
is of the order $(\ln (M+2))^d$ \cite{cheney2019}.
 This together with
the exponential convergence
of Chebyshev polynomial approximation, see \cite[Theorem 8.2]{trefethen2013} and \cite[Theorem 2.2]{xiang2012}, implies that
\vspace{-.6em}\begin{equation}\label{exponentialapproximationerror.interpolation}
\tilde b_M:=\sup_{{\bf t}\in  [{\pmb \mu}, {\pmb \nu}]}
|1-h({\bf t})C_M({\bf t})|\le D_1 r_1^M, \ M\ge 0,
\vspace{-.4em}\end{equation}
for some  positive constants
  $D_1\in (0, \infty)$ and $r_1\in (0, 1)$.

  Presented at the bottom row of Table
  \ref{MaxAppErr.tab}
  %  \ref{CompNumeric.Table}
      is  the  maximal approximation error $\tilde b_M, 0\le M\le 4$, of
       the Chebyshev interpolation polynomial $C_M$ on $[0, 2]$, ChebyInt for abbreviation,  to the reciprocal of the univariate function 
\vspace{-.6em}\begin{equation}\label{h1.def}
h_1(t) = (9/4-t)(3 + t), \ t\in [0, 2]
\vspace{-.4em}\end{equation}
in \cite[Eqn. (5.4)]{ncjs22}. 
%Shown at
%   the bottom right of Figure \ref{approximation2.fig}
%   is the  exponential    approximation
%  property of the Chebyshev interpolation polynomial $C_M$  to the reciprocal of the bivariate  function $h_{\gamma_1, \gamma_2}, 0\le \gamma_1, \gamma_2\le 2$, see also
%  the numerical simulations in Section \ref{circulantgraph.demo}.
We observe that the Chebyshev interpolation polynomial approximation outperforms the Jacobi polynomial approximations  with $\alpha = \beta = 1/2$ in \cite{cong2021} and the Chebyshev polynomial approximation.

\begin{table}[t]
		\renewcommand\arraystretch{1.2}
		\centering
		\caption{ Shown are the maximal approximation errors of Jacobi polynomial approximations, Chebyshev polynomial approximation and Chebyshev interpolation polynomial approximation to $1/h_1$ on $[0, 2]$ with the polynomial degree $0\le M\le 4$.
	}
		%\label{CirculantGraphCPA.Table}			
       \begin{tabular} {|c|c|c|c|c|c|c|c|c|c|}			
			\hline
			\backslashbox{$(\alpha, \beta)$}%{ME}
{M}& 0 & 1 & 2  & 3
& 4  %& 5  & 6
   \\
			\hline
Cheby. Poly.  &  1.0463  &   0.5837  &   0.2924    &   0.1467
&   0.0728  %&   0.0367   &    0.0184
% &    0.0184    0.0092    0.0046    0.0023     0.0011    0.0006
\\			 \hline
			(1/2,  1/2)   &
    0.7014  &  0.5904  &    0.3897  &     0.2505
     &     0.1517   %&  0.0893     & 0.0513
    %   0.0513    0.0290    0.0162    0.0089    0.0049    0.0027
    \\  \hline
%			(0, 0) &
%    0.7409   &  0.6153  &   0.3667   &   0.2146
%    &   0.1202  %&     0.0660     &0.0357
%    %   0.0357    0.0191    0.0102    0.0054    0.0028     0.0015
% \\			 \hline
%%	(1, 1) &     0.7140  &   0.5626  &    0.3927   &  0.2686
%%&   0.1720  % &    0.1066  &0.0643
%%%  0.0643    0.0380    0.0220    0.0126    0.0072    0.0040
%%\\			 \hline
%%			(-.5, .5) &    1.8612 &     1.8855  &    1.3522    &    0.8937
%%&    0.5534   %&   0.3304  & 0.1920
%%%  0.1920    0.1094    0.0614    0.0340    0.0187    0.0102
%% \\			 \hline
%(.5, -.5) &    0.7720  &    0.5603  &    0.3563   &    0.2184
%&    0.1289  %&   0.0744 &0.0422
%%   0.0422    0.0236    0.0130    0.0071    0.0039    0.0021
% \\			 \hline
%(0, -.5) &
%    0.7356 &    0.4760  &    0.2749  &    0.1548
%    &  0.0850  %&    0.0460     & 0.0246
%    %  0.0246    0.0131    0.0069    0.0036    0.0019    0.0010
%% \\			 \hline
%%(0, -.8)  &    0.7787  &   0.4993   &  0.2927   &    0.1512   &    0.0879  &   0.0470
%%%   0.0451    0.0449    0.0456    0.0465    0.0472    0.0480
%  \\			 \hline
  {ChebyInt} &
  0.7500   &  0.4497 &     0.2342  &     0.1186
  &     0.0595  %&    0.0298    &0.0149
  %   0.0074    0.0037    0.0019    0.0009
    \\ \hline
		\end{tabular}
		\label{MaxAppErr.tab}
		\vspace{-2em}
	\end{table}

\section{Chebyshev interpolation approximation algorithm}\label{cipa.sec}
Let 
 $h$ be a multivariate polynomial   satisfying
\eqref{polynomial.assump}, and
$C_M, M\ge 0$, be the Chebyshev interpolation polynomial approximation to $1/h$
in  \eqref{chebyshevinterpolation.def}.
 Set ${\bf H}=h({\bf S}_1, \ldots, {\bf S}_d)$ and 
 ${\bf C}_M=C_M({\bf S}_1, \ldots, {\bf S}_d), M\ge 0$.
By the  spectral assumption \eqref{jointspectralcubic.def}, % on commutative graph shifts ${\bf S}_1, \ldots, {\bf S}_d$,
%we have that
the spectral radii of ${\bf I}-{\bf C}_M{\bf H}$
 are bounded by  $\tilde b_M$  in \eqref{exponentialapproximationerror.interpolation}
respectively, i.e.,
\vspace{-.6em}\begin{equation}
\rho({\bf I}-{\bf C}_M{\bf H})\leq \tilde b_M,
 \ M\ge 0.
\vspace{-.4em}\end{equation}
Therefore with appropriate selection of the polynomial degree $M$, we obtain the exponential convergence of  the following iterative polynomial approximation algorithm for inverse filtering,
\vspace{-.6em}\begin{equation} \label{cipa.def}  %\label{jacobiapproximation.eqa}
\left\{\begin{array}{l}
{\bf e}^{(m)} = {\bf H} {\bf  x}^{(m-1)} -  {\bf y}\\
{\bf x}^{(m)} = {\bf x}^{(m-1)} -{\bf C}_M {\bf e}^{(m)}, \ m\ge 1
\end{array}
\right.
\vspace{-.6em}\end{equation}
with  arbitrary initials
${\bf x}^{(0)}$,
where 
 the input ${\bf y}$
   is obtained via the filtering procedure \eqref{filtering.def}.

\begin{theorem}\label{exponentialconvergence.cipa.thm}
Let ${\mathbf S}_1, \ldots, {\bf S}_d$ be commutative graph shifts satisfying
\eqref{jointspectralcubic.def},
  $h$ be a multivariate polynomial satisfying \eqref{polynomial.assump},
and let  $\tilde b_M$ be given in \eqref{exponentialapproximationerror.interpolation}.
If $ \tilde b_M<1$, 
%\begin{equation}
% \tilde b_M<1,\end{equation}
 then for any  input ${\bf y}$, the sequence
${\bf x}^{(m)}, m\ge 0$, in the iterative algorithm \eqref{cipa.def} 
converges to  the output ${\bf H}^{-1}{\bf y}$ of the inverse filtering procedure \eqref{inverseprocedure} exponentially. In particular,
for any $r\in (\rho({\bf I}-{\bf C}_M{\bf H}), 1)$
there exists a positive constant $C$ such that
\begin{equation}
\| {\bf x}^{(m)}- {\bf H}^{-1}{\bf y}\|_2 \le C \|{\bf x}^{(0)}- {\bf H}^{-1}{\bf y}\|_2 r^m, \ m\ge 0.
\end{equation}
\end{theorem}

\begin{proof}  
As ${\mathbf S}_1, \ldots, {\bf S}_d$ are commutative graph shifts, $1/h$ is analytic on the joint spectrum  $[{\pmb \mu}, {\pmb \nu}]$, we have  
\begin{eqnarray}\label{xm.nondef}{\bf x}^{(m)}-{\bf H}^{-1}{\bf y}&=&({\bf I}-{\bf C}_M{\bf H})^m({\bf x}^{(0)}-{\bf H}^{-1}{\bf y}), \nonumber\\
&-& \sum_{n=m+1}^{\infty} ({\bf I}-{\bf C}_M{\bf H})^n{\bf C}_M{\bf y})\end{eqnarray}
from the iterative algorithm \eqref{cipa.def}. 
By the  spectral assumption \eqref{jointspectralcubic.def} on commutative graph shifts ${\bf S}_1, \ldots, {\bf S}_d$, 
%we have that
the spectral radii of ${\bf I}-{\bf C}_M{\bf H}$ 
 is bounded by $\tilde{b}_M$ in \eqref{exponentialapproximationerror.interpolation}, i.e.,
\begin{equation}
\rho({\bf I}-{\bf C}_M{\bf H})\leq {\tilde b}_M. 
\end{equation} 
 By Gelfand's formula on spectral radius, there exists a positive constant $C$ for any $r\in (\rho({\bf I}-{\bf C}_M{\bf H}), 1)$ such that
\begin{equation}  \label{xm.def22}
\|({\bf I}-{\bf C}_M{\bf H})^m\|_2\le C r^m,\  n\ge 1.
\end{equation}
From  \eqref{xm.nondef} and \eqref{xm.def22},  it follows that
\begin{eqnarray*} \label{xm.nondef2}
  \|{\bf x}^{(m)}-{\bf H}^{-1}{\bf y}\|_2&\le& \|({\bf I}-{\bf C}_M{\bf H})^m\|_2 \|{\bf x}^{(0)}-{\bf H}^{-1}{\bf y}\|_2 \nonumber\\
	& \le& C r^m   \|{\bf x}^{(0)}-{\bf H}^{-1}{\bf y}\|_2, m\ge 0. 
\vspace{-0.4em}
\end{eqnarray*}
This proves the exponential convergence of ${\bf x}^{(m)}, m\ge 0$.  
\end{proof}

 We call the iterative polynomial approximation
  algorithm \eqref{cipa.def}  
  as  {\em  Chebyshev interpolation polynomial approximation algorithm}, CIPA for abbreviation. We remark that in 
each iteration in  CIPA  contains essentially two filtering procedures associated
with polynomial filters  ${\bf C}_M$  and ${\bf H}$,
and hence it can be implemented at the  vertex level with one-hop communication,  see
Algorithm \ref{CIPA.algorithm}. % for the  distributed implementation of more general edge-variant FIR graph filters.
Therefore the CIPA algorithms
%the iterative Jacobi polynomial approximation algorithm and  the iterative  Chebyshev interpolation polynomial approximation algorithm
   can be implemented  on a distributed network with each agent equipped with limited storage and data processing ability,  and one-hop communication subsystem.
  More importantly,  the memory, computational cost and communication expense for
each agent of the network are {\bf independent} on the size of the whole network.

  \begin{algorithm}[t]
\caption{The CIPA algorithm  to implement the inverse filtering procedure ${\bf y}\longmapsto {\bf H}^{-1}{\bf y}$
  at a vertex $i\in V$. }
\label{CIPA.algorithm}
\begin{algorithmic}  %[1]

\STATE {\bf Inputs}: Polynomial coefficients of polynomial filters ${\bf H}$ and ${\bf C}_M$,  entries $S_k(i,j), j\in {\mathcal N}_i$ in the $i$-th row of the shifts ${\bf S}_k, 1\le k\le d$,
the  value $y(i)$  of the input signal ${\bf y}=(y(i))_{i\in V}$ at the vertex $i$, and number $m$ of iteration.

%\STATE {\bf Operation}: Evaluate $m_k=\mu(B(k, r))$, compute ${\bf F}_k= {\bf H}_{0,k}^T{\bf H}_{0,k}+ {\bf H}_{1,k}^T{\bf H}_{1,k}$,
% find its inverse  $({\bf F}_k)^{-1}$, and then compute $ {\bf G}^L_{l; k}:=({\bf F}_k)^{-1} {\bf H}_{l,k}^T, l=0, 1$.

%  , and a local approximation
%${\bf G}_k=(\tilde f_k(i,j))_{i,j\in B(k, 2r)}$  to the matrix ${\bf H}$
% and compute ${\bf F}_k= {\bf H}_{0,k}^T{\bf H}_{0,k}+ {\bf H}_{1,k}^T{\bf H}_{1,k}$ and
%$({\bf F}_k)^{-1}=(\tilde f_k(i,j))_{i,j\in B(k, 2r)}$

\STATE {\bf Initialization}:  Initial $e^{(0)}(i)=y(i)$, $x^{(0)}(i)=0$ and $n=0$.

\STATE{\bf Iteration}:  Use the iteration in  \cite[Algorithm 4] {ncjs22}  except replacing  ${\widetilde {\bf G}}_L$  by
 ${\bf C}_M$, and the output is
$ x^{(n)}(i)$.

\STATE {\bf Output}: The approximated  value $ x(i)\approx x^{(m)}(i)$  is  the output signal ${\bf H}^{-1}{\bf y}=( x(i))_{i\in V}$ at the vertex $i$.
\end{algorithmic}  
\end{algorithm}

%\begin{remark} \label{CIPAalgorithm.remark}
%{\em   We remark that each iteration in the CIPA algorithm \eqref{Chebysheviterativedistributedalgorithm.eqn1}
%can be implemented at vertex level, 
%In  each iteration of the CIPA algorithm,  every vertex in a distributed network 
%needs about $O((M+1)^{d-1}+ \prod_{k=1}^{d-1} (l_k+1))$ steps with each step
%containing data exchanging among adjacent vertices and  weighted linear combination
%of values at adjacent vertices.
%The memory requirement for each agent is about
%$O\big( (\deg {\mathcal G}+L_d+1) \prod_{k=1}^{d-1}(L_k+1)+ (\deg {\mathcal G}+K+1)(K+1)^{d-1}\big)$.
%The total  operations of addition and
%multiplication to implement each iteration of Algorithm \ref{ICPA.algorithm}
%in a distributed network  and
%to implement \eqref{Chebysheviterativedistributedalgorithm.eqn1} in a central facility are almost
%the same, which are both about $O\big(N (\deg {\mathcal G}+1) (\prod_{k=1}^{d}(L_k+1)+(K+1)^{d})\big)$.
%}	
%\end{remark}
%
\section{Numerical Experiments}\label{num.sec}

Circulant graphs are widely used in image processing
\cite{ekambaram13, vnekambaram13, ncjs22, dragotti19, dragotti19a}.
 Our numerical results  show that
  the CIPA  %Chebyshev interpolation polynomial approximation algorithms CIPA and
%Jacobi polynomial approximation algorithms
 have impressive performances to implement the inverse filtering procedure than
the Chebyshev polynomial approximation algorithm  in \cite{ncjs22} and
the gradient descent method  in \cite{Shi15} do.  Some Tikhonov regularization problem
can be converted to an inverse filtering procedure  \cite{ncjs22, Grassi2018}, we  also demonstrate the denoising performance of the polynomial approximation algorithms
 to the walking dog dataset.

 \vspace{-1mm}
\subsection{Polynomial approximation algorithms}
\label{circulantgraph.demo}
 \vspace{-1mm}
% \subsubsection{Signals on circulant graphs}
	Let $N\ge 1$ and we say that
 $a=b\ {\rm mod }\ N$ if $(a-b)/N$ is an integer.
 The  {\em circulant graph} ${\mathcal C}(N, Q)$  generated by  $Q=\{q_1, \ldots, q_L\}$
is a simple graph with  the vertex set   $V_N=\{0, 1, \ldots, N-1\}$  and the edge set
%\begin{equation}\label{circulant.edgedef}
$E_N(Q)=\{(i, i\pm q\ {\rm mod}\ N),\  i\in V_N, q\in Q\}$,  % \end{equation}
where   $q_l, 1\le l\le L$, are integers contained in $[1, N/2)$.
 Let  $Q_0 = \{1, 2, 5\}$ and 
 the polynomial filters  be  ${\bf H}_1=h_1( {\bf L}_{\mathcal{C}(N,Q_0)}^{\rm sym})$, 
   the input signal ${\bf x}$ have i.i.d. entries randomly selected in $[-1, 1]$,
    and  the  input signal  ${\bf y}={\bf H}_1{\bf x}$, 
 where $h_1(t)=(9/4-t)(3+t)$ in \eqref{h1.def}, and ${\bf L}_{\mathcal{C}(N,Q_0)}^{\rm sym}$ is the symmetric normalized Laplacian  on the circulant graph  $\mathcal{C}(N,Q_0)$.  Shown in Table
\ref{CirculantGraphCPA.Table} are
averages of the relative  iteration error
 \vspace{-.3em}\begin{equation*} \vspace{-.4em}{\rm E}(m)=\frac{ \|{\bf x}^{(m)}-{\bf x}\|_2}{\|{\bf x}\|_2},\  m\ge 1, \end{equation*} 
over 1000 trials to implement the inverse filtering procedure ${\bf y}\longmapsto {\bf H}_1^{-1} {\bf y}$ via CPA ( the Chebyshev
polynomial approximation algorithm in \cite{ncjs22}), 
 the  JPA($1/2, 1/2$) (Jacobi  polynomial approximation with appropriate selection of parameters $\alpha=1/2$ and $\beta=1/2$ in \cite{cong2021}), CIPA with zero initial  ${\bf x}^{(0)}={\bf 0}$, the  gradient descent method  with optimal step size in \cite{Shi15} and  autoregressive moving
average method in  \cite{Leus17},
OGDA and ARMA for abbreviation, 
where ${\bf x}^{(m)}, m\ge 1$, are the output of the polynomial approximation algorithm
 at $m$-th iteration and   $M$ is  the degree of polynomials the polynomial approximation. 
We observe %from Table \ref{CirculantGraphCPA.Table}
that
%Chebyshev interpolation polynomial approximation algorithms
CIPA
 have the  {\bf best} performances on the implementation of inverse filtering procedure
than the JPA($1/2,1/2$) in \cite{cong2021}, 
% Chebyshev polynomial approximation algorithm
 CPA in \cite{ncjs22} does,
and
CIPA  %the proposed Chebyshev interpolation polynomial approximation algorithms
%CIPA and
%Jacobi polynomial approximation algorithms  JPA($\alpha, \beta$)
has much better performance  than  than
the  gradient descent method  does.   % the approximation filter is selected.

%As the filter ${\bf H}_1$ is a positive definite matrix,  the inverse filtering procedure ${\bf y}\longmapsto {\bf H}_1^{-1}{\bf y}$
%can also be implemented by 
%Shown in the sixth and seventh rows of  Table \ref{CirculantGraphCPA.Table} are the relative  iteration errors to implement the inverse filtering ${\bf y}\longmapsto {\bf H}_1^{-1}{\bf y}$.
%It indicates that  the CIPA  %Chebyshev interpolation polynomial approximation algorithms CIPA
%has {\bf superior} performances to implement the inverse procedure than
%the  gradient descent method  does. %\st{ even as observed in} \cite[Table 1]{ncjs22} \st{that
%% it has  better performance than the CPA does, see the relative iteration error listed in the first row of
%%  Table} \ref{CirculantGraphCPA.Table}.  %superierhas better performance than the
%

\begin{table}[t]
		\centering
		\caption{ 		Average   relative iteration errors  $E(m)$
to implement the inverse filtering  ${\bf y}\longmapsto {\bf H}_1^{-1} {\bf y}$ on the circulant graph ${\mathcal C}(1000, Q_0)$ via polynomial approximation algorithms with polynomial degree $M=1$,
 the gradient descent algorithm with optimal step size and ARMA, where we take zero as the initial.}
		\label{CirculantGraphCPA.Table}			
       \begin{tabular} {|l|c|c|c|c|c|c|c|}	
   			\hline
   \hline
			\backslashbox{Alg.} %orithm} %$(\alpha, \beta, M)$}
%{ME}
{Iter. $m$} %ation}
& 1 & 2 & 3 & 4 & 5 \\ % & 6 &7&8&9 &10 &12&16&20 \\
%			\hline
%			\hline
%     \multicolumn{6}{c}{\multirow{1}{*}{$M=0$}}\\
%     \hline
%           \hline
%% {$\eta$=3/4, ISNR=  3.3755} %\vline
%%\hline
%CPA	&
%0.5686  &  0.4318  &  0.3752  &  0.3521 &   0.3441\\
%%  &  0.3422   & 0.3449&  0.3500  &  0.3565  &  0.3643 &   0.3828  &  0.4296 &   0.4890\\
%\hline
%JPA(${1}/{2}$, ${1}/{2}$)	&	
%0.3007 	&   0.1307  	&  0.0677	&    0.0379  	&  0.0219  \\
%%	&  0.0129  	&  0.0076		&  0.0045   &  0.0027   &  0.0016  &   0.0006  &   0.0001 &    0.0000\\
%\hline
%JPA(${1}/{2}$,-${1}/{2}$)&
%0.2298 &   0.0955&    0.0452 &   0.0223   & 0.0113 \\
%%  &  0.0057  &  0.0029& 0.0015 &  0.0008 &   0.0004  &  0.0001  &  0.0000  &  0.0000\\
%\hline
%JPA(0,-${1}/{2}$)&
%   0.2296  &  0.0833 &   0.0337 &   0.0141  &  0.0060 \\
%   %&   0.0026    &0.0011&   0.0005  &  0.0002  &  0.0001  &  0.0000 &   0.0000  &  0.0000 \\
%\hline	
%CIPA  &
% 0.2189  &  0.0822 &  0.0347  &  0.0154  &  0.0070 \\
% \hline
% % &  0.0033   & 0.0015& 0.0007  &  0.0003 &   0.0002  &  0.0000 &   0.0000 &   0.0000\\
%% ARMA & 0.3259 & 0.2583 & 0.1423 & 0.1098 & 0.0718\\
%% \hline
%OGDA & 0.2350 & 0.0856 & 0.0349 & 0.0147 & 0.0063\\
% \hline
% ARMA & 0.3259&    0.2583 &  0.1423 &   0.1098   & 0.0718\\
  \hline
% 			\hline
%     \multicolumn{6}{c}{\multirow{1}{*}{$M=1$}}\\ \hline
%           \hline
CPA&
0.4494   & 0.2191 &   0.1103  &  0.0566   & 0.0295 \\
%&   0.0155  &  0.0082& 0.0044 &   0.0024  &  0.0013 &   0.0004   & 0.0000  &  0.0000\\
\hline
 JPA(${1}/{2}$, ${1}/{2}$)
	&  0.2056  &  0.0769 &   0.0390 &   0.0213 &   0.0119\\
% &    0.0067 &   0.0038&   0.0022   & 0.0012  &  0.0007   & 0.0002 &   0.0000  & 0.0000\\
\hline
%JPA(${1}/{2}$, -${1}/{2}$)
%& 0.1624   & 0.0297 &   0.0056 &   0.0011 &   0.0002 \\
%% &  0.0000  &  0.0000&  0.0000  &  0.0000 &   0.0000  &  0.0000   & 0.0000   & 0.0000
%\hline
%JPA(0, -${1}/{2}$)
%&  0.2580  &  0.0754 &   0.0225   & 0.0068  &  0.0021  \\
%% & 0.0006   & 0.0002& 0.0001  &  0.0000  &  0.0000  &  0.0000 &   0.0000  &  0.0000
%\hline
CIPA
&0.2994 &    0.1010  &   0.0349   &  0.0122  &   0.0043  \\
%  & 0.0015   &  0.0005&  0.0002 &    0.0001   &  0.0000 &    0.0000   &  0.0000  &   0.0000\\
\hline
OGDA & 0.2350 & 0.0856 & 0.0349 & 0.0147 & 0.0063\\
 \hline
 ARMA & 0.3259&    0.2583 &  0.1423 &   0.1098   & 0.0718\\
   \hline
\end{tabular}
 \vspace{-3em}\end{table}

\subsection{Denoising dancing dog dataset}
 
In the second simulation, we consider
 applying polynomial approximation algorithms to denoise the walking dog dataset ${\bf W}$ of size $ 442854=2502 \times 59 \times 3$ \cite{Grassi2018}.
 %Cartesian product graph is widely used to model the time-varying graph signal.
Let ${\mathcal T}$ be the line graph  with 59 vertices and
 ${\mathcal W}=(V, E)$ be the undirected graph with  $2502$ vertices and
 edges  constructed by the  5 nearest neighboring algorithm.
 The walking dog data  is modelled as a time-varying signal ${\bf W}(t, i)\in {\mathbb R}^3,  t\in \{1, \ldots, 59\}, i\in V$
on the product graph ${\mathcal T}\times {\mathcal W}$.
Consider the scenario that the known dataset is the noisy walking dog dataset
 \vspace{-.6em}\begin{equation}\label{walkingdog.model}
\widetilde {\bf W}={\bf W}+\lambda {\pmb \eta} \vspace{-.4em}\end{equation}
 corrupted
by some random noises $\lambda {\pmb \eta}$, where ${\pmb \eta}$ has its
 components ${\pmb \eta}(t, i), t\in \{1, \ldots, 59\},  i\in V$
being independently and randomly selected with a normal Gaussian distribution,
and the normalization factor
$\lambda= {0.2 \|{\bf W}\|_F} ({{\mathbb E}\|{\pmb \eta}\|_F^2})^{-1/2}=29.1398$ is so chosen that the Frobenius norm  $\lambda \|{\pmb \eta}\|_F$ of the additive noise is about 20\% of the norm $\|{\bf W}\|_F$ of
 the walking dog dataset.  %{\color{red} Delete: In particular, we take $\lambda= {0.2 \|{\bf w}\|_F}/{\sqrt{442854}}\approx 29.1398$
%  as ${\mathbb E}\|{\pmb \eta}\|_F^2= 442854$ by the i.i.d  normal Gaussian distribution property for its component.}

A conventional denoising procedure is  the  Tikhonov denoising approach in each coordinate dimension (\cite{ncjs22, Grassi2018}).
Then the denoised walking dog dataset is given by
	 \vspace{-.3em}\begin{equation}\label{Dogminimization0}
	\widehat {\bf w} := {\rm arg}\min_{\bf z}  \|{\bf z}-{\widetilde {\bf w}}\|_2^2+\gamma_1 {\bf z}^T  {\bf S}_1 {\bf z} + \gamma_2 {\bf  z}^T  {\bf S}_2 {\bf z},
	 \vspace{-.3em}\end{equation}
 where $\widetilde {\bf w}$ is  the vectorization of the  noisy dog dataset $\widetilde {\bf W}$,  $ {\bf S}_1={\bf I}\otimes {\bf L}^{\rm sym}_{\mathcal{W}}$, $ {\bf S}_2={\bf L}_{\mathcal T}^{\rm sym}\otimes {\bf I}$,  ${\bf L}^{\rm sym}_{\mathcal W}$ and ${\bf L}_{\mathcal T}^{\rm sym}$ are symmetrically normalized Laplacian matrices on the graph ${\mathcal W}$ and ${\mathcal T}$ respectively,
 and penalty constants $ \gamma_1, \gamma_2\ge 0$ are used to balance the fidelity and regularization in the vertex and temporal domains. 
  %Similar as in \cite{Grassi2018}, we implement the  inverse filter procedure  $\widetilde {\bf W}\longmapsto \widehat {\bf W}$
 Shown in the top row of Figure \ref{denoisecompdog1.fig} are the snapshots of  the original walking dog dataset, the noisy  walking dog dataset
  and the denoised walking dog dataset at time $t=1$, where the penalty constants $\gamma_1=\gamma_2=1$ are used.   We observe
  that the proposed iterative polynomial approximation algorithms can effectively denoise the walking dog dataset.
 From the top right and top middle plots of Figure \ref{denoisecompdog1.fig}, we see that
 the denoised dog dataset reveals the shape of the dog, while
the noisy walking dog dataset obscures the gesture of the two front legs.

 	\begin{figure}[t] %[h]
\centering
\includegraphics[width=28mm, height=24mm]
{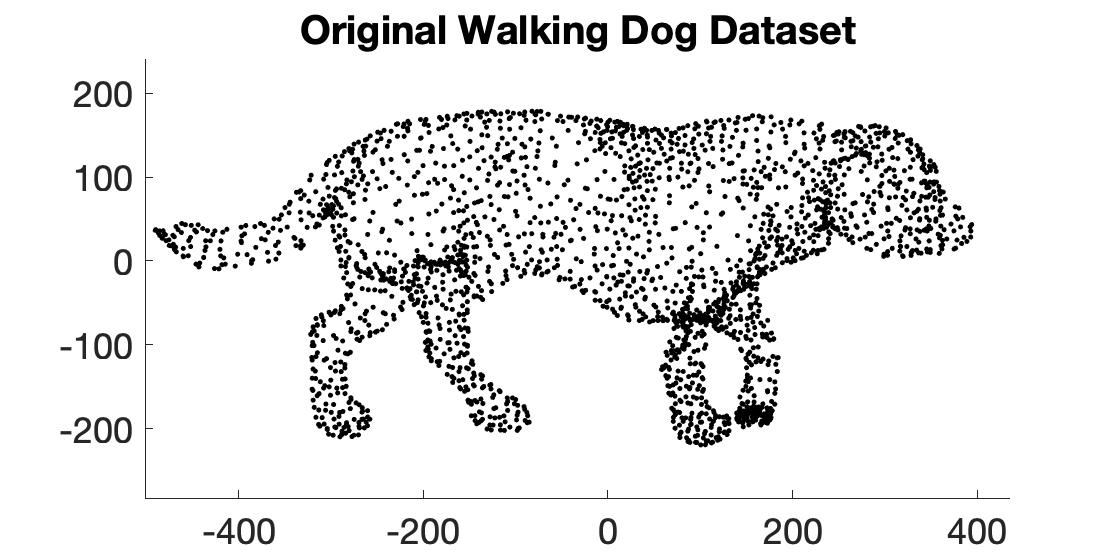}
\includegraphics[width=28mm, height=24mm]
{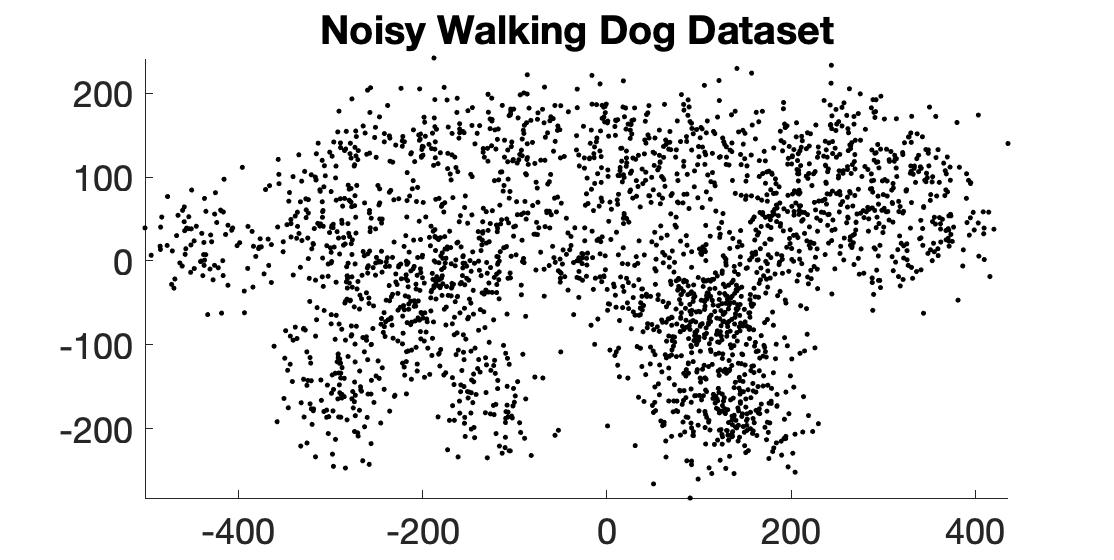}
\includegraphics[width=28mm, height=24mm]
{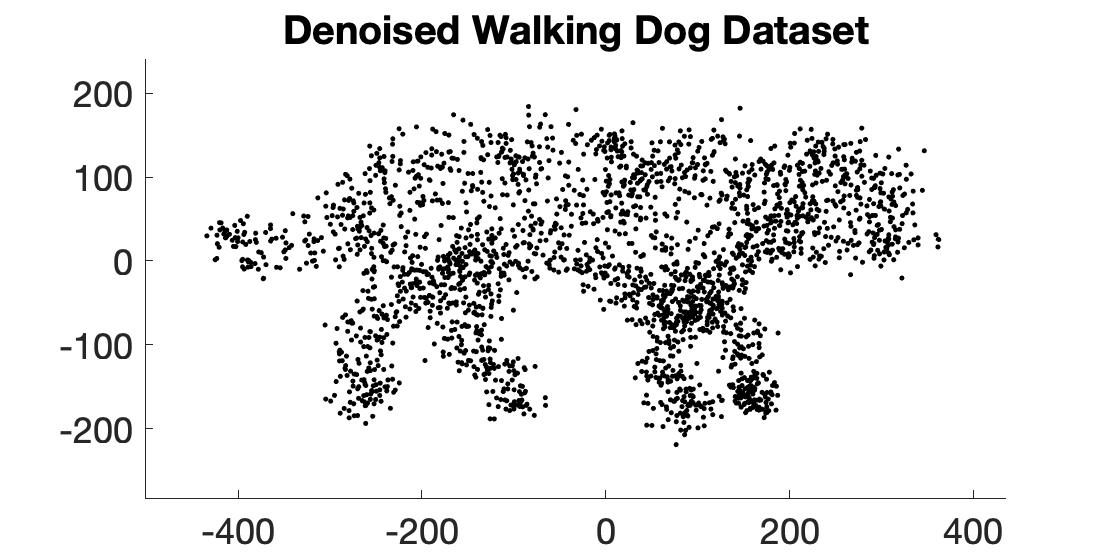}\\
\includegraphics[width=28mm, height=24mm]
{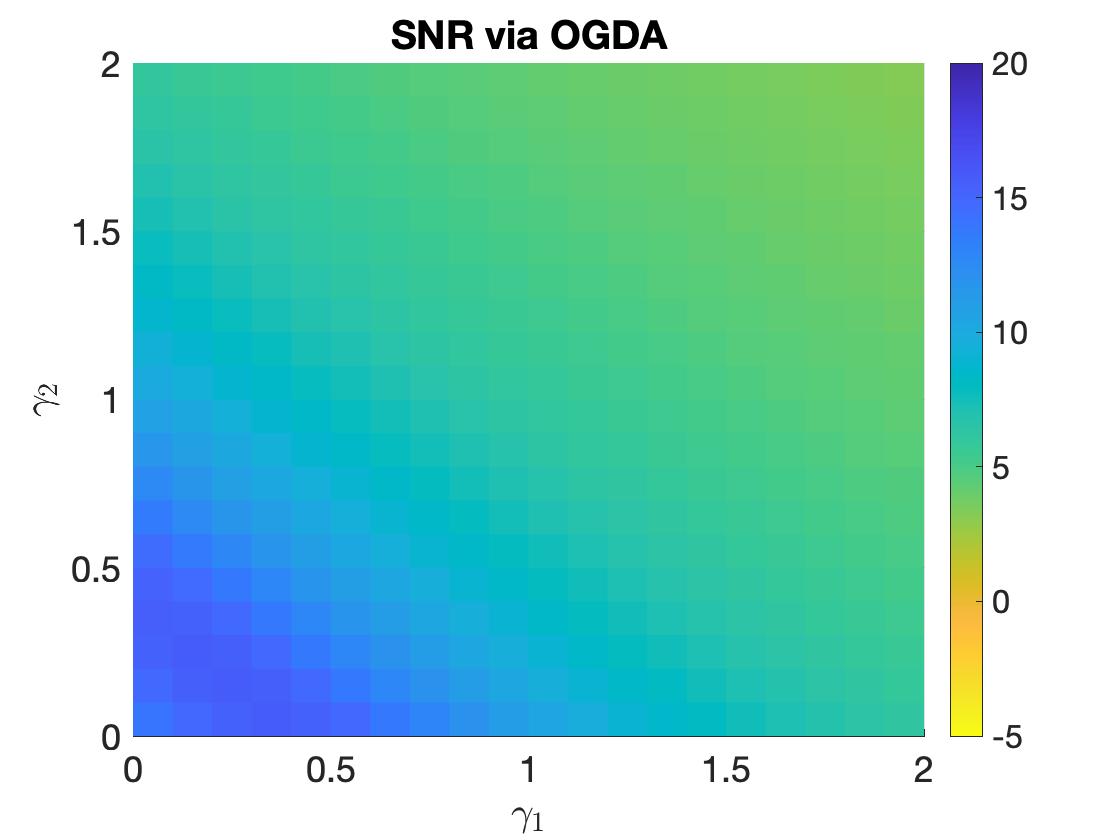}
\includegraphics[width=28mm, height=24mm]
{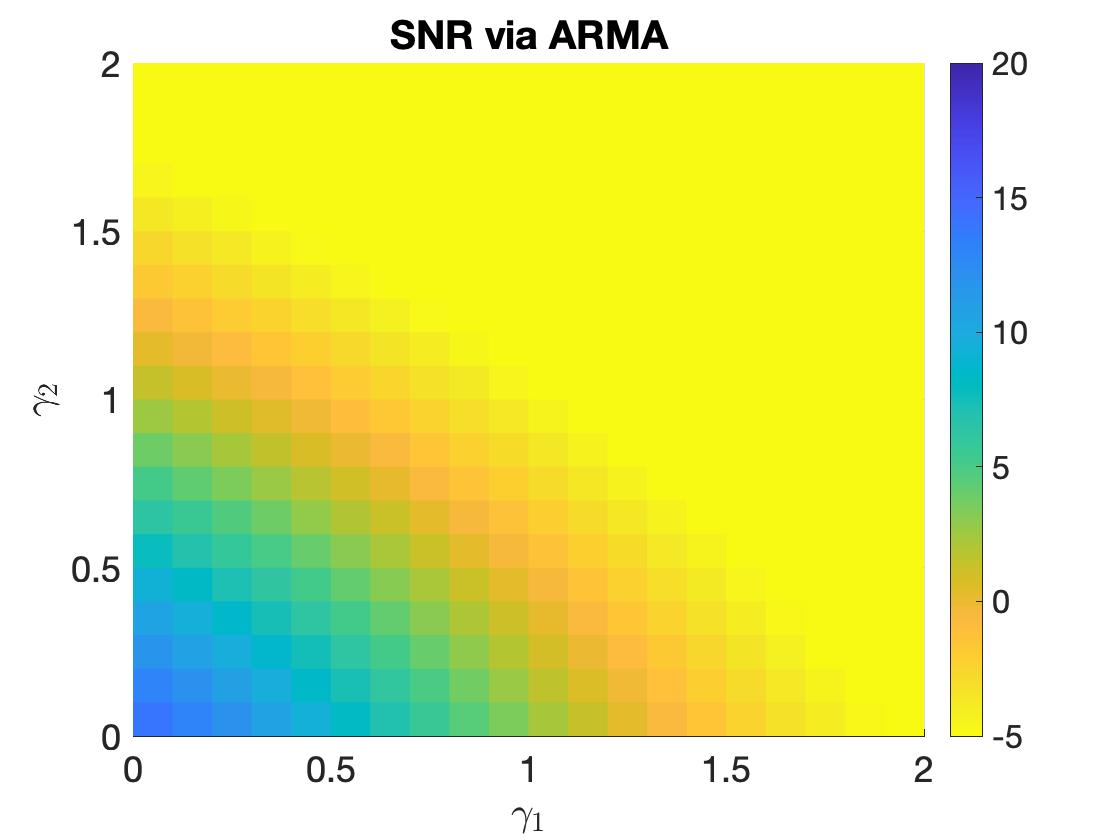}
\includegraphics[width=28mm, height=24mm]
{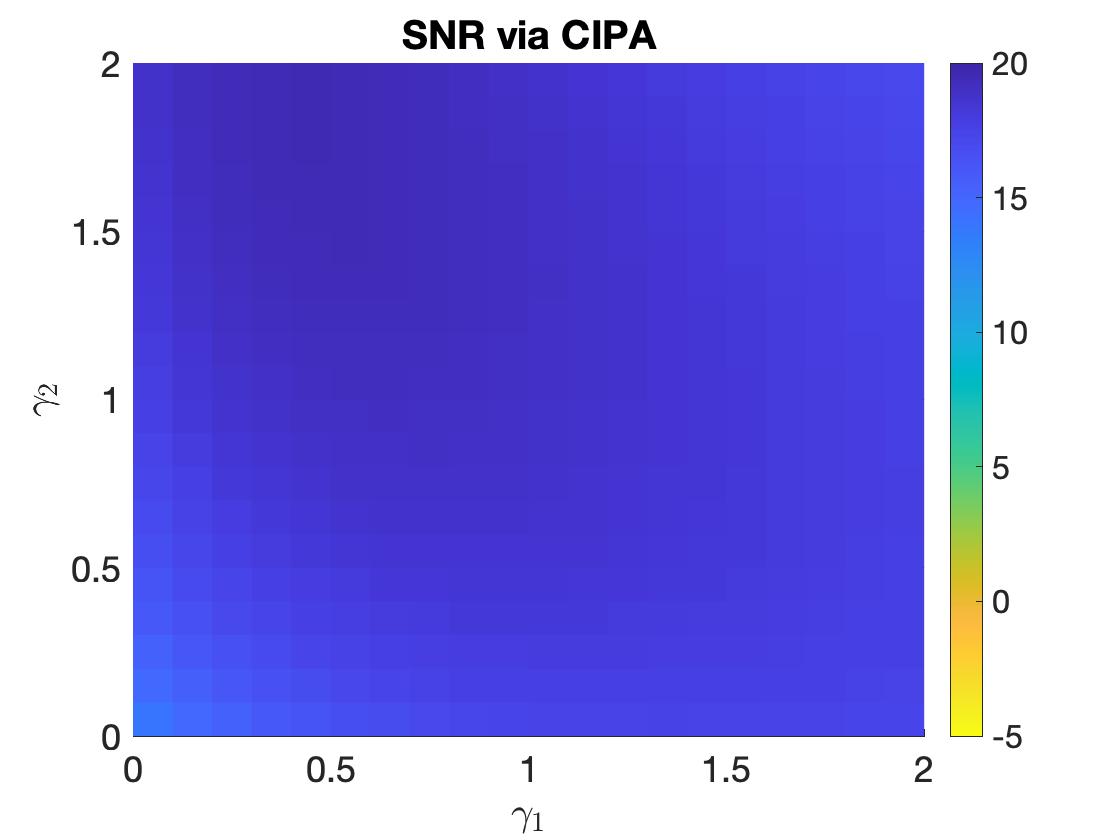} 
\caption{Plotted on the top left, top middle and top right are  the snapshots of  the original walking dog  dataset,
 the noisy dataset and the denoised dataset at the time $t=1$, where  the input SNR is  $13.9734$ and the output SNR is $18.9654$. Plotted on the bottom  from the left to the right are average output  signal-to-noise ratios $\max({\rm SNR}(m), -5)$ at $m$-th iteration of OGDA, ARMA and CIPA
%, JPA with $(\alpha, \beta)$ being $(-0.5, -0.5)$, $(0.5, 0.5)$, $(0, 0)$
on denoising the walking dog dataset through Tikhonov denoising approach in \eqref{Dogminimization0} with respect to different penalty constants $\gamma_1, \gamma_2\in [0, 2]$, where $m=3$.
}
\label{denoisecompdog1.fig}
 \vspace{-2em}\end{figure}

%  Set
% $\tilde {\bf S}_1={\bf I}\otimes {\bf L}^{\rm sym}_{\mathcal{W}}$, $\tilde {\bf S}_2={\bf L}_{\mathcal T}^{\rm sym}\otimes {\bf I}
% $
% and
% ${\bf F}_{\gamma_1, \gamma_2}={\bf I}+\gamma_1 \tilde {\bf S}_1+ \gamma_2 \tilde {\bf S}_2, \
% \gamma_1, \gamma_2\ge 0$.
% One may verify that
%  the explicit solution to the minimization problem  \eqref{Dogminimization0} is
%$$\widehat {\bf W}= ({\bf F}_{ \gamma_1,  \gamma_2})^{-1} {\widetilde
%{\bf W}},$$
% and that
% the proposed approach to  denoise  the walking dog dataset becomes an inverse filtering procedure \eqref{inverseprocedure}
%with ${\bf H}$   and ${\bf b}$ replaced by ${\bf F}_{\gamma_1,\gamma_2}$ and $\widetilde {\bf W}$ respectively.
% We compare the performance of the Tikhonov denoising approach through
%the CIPA  algorithm with degree $M=1$,
%	the CPA  algorithm in \cite{ncjs22} with degree $M=1$ (JPA with  $(\alpha, \beta)$ being $(-1/2, -1/2)$), the JPA algorithm with
%degree $M=1$ and $(\alpha, \beta)$ being $(1/2, 1/2)$  (denoted by JPA1) and $(0, 0)$ (denoted by JPA2), % gradient descent with optimal step size \eqref{gd0.def} in \cite{Shi15}
% OGDA  in \cite{Shi15},
%%the joint graph Fourier transform method  used in \cite{Grassi2018}
%and the ARMA model proposed in \cite{Leus17}.
Denote the output of the $m$-th iteration of the proposed algorithms by
$\widehat{\bf w}^{(m)}$ and define
the output signal-to-noise ratio of the proposed algorithms at $m$-th iteration by
	\vspace{-0.5em}$${\rm SNR}(m)=-20 \log_{10} \frac{\|\widehat{\bf w}^{(m)}-{\bf w}\|_2}{\|{\bf w}\|_2}, \ m\ge 1.$$
Comparing with CIPA,  OGDA  in \cite{Shi15} has a slower convergence rate, see the bottom left plot of Figure \ref{denoisecompdog1.fig}. Our experiments also indicate that the OGDA may achieve a similar denoising performance after 20 iterations  to that CIPA do after 3 iterations.
Plotted in the  bottom middle  of Figure \ref{denoisecompdog1.fig} are
the average output signal-to-noise ratios
$\max({\rm SNR}(3), -5)$ of
the ARMA model proposed in \cite{Leus17}
with respect to the penalty constants $\gamma_1, \gamma_2\in [0, 2]$
over 20 trials on the random noise ${\pmb \eta}$ in \eqref{walkingdog.model}.
It shows that the ARMA model  may fail to denoise the walking dog dataset for  penalty constants $\gamma_1, \gamma_2$ not close to zero.
Recall that the requirement for convergence of the ARMA model
 in \cite{Leus17} is that the spectrum of
 $ {\bf T}:=\gamma_1  {\mathbf S}_1+\gamma_2 {\mathbf S}_2$
 is contained in $(-1, 1)$.
Observe that the spectrum of
  $ {\bf T}$
  is contained in $[0, 2(\gamma_1+\gamma_2)]$
from the spectral properties of
${\mathbf S}_1$ and ${\mathbf S}_2$.
Then a possible explanation for why the ARMA model did not perform well is that it does not meet the requirement
for the convergence of the ARMA method
when $\gamma_1, \gamma_2$ are  not close to zero.

{\bf Acknowledgement}: This work is partially supported by National Key RD Program of China (No. 2024YFA1013703),  %the National Science Foundation (DMS-1816313),
National Nature Science Foundation
of China (12171490), Guangdong  Basic and Applied Basic Research
Foundation (2022A1515011060), and Fundamental Research Funds for the Central Universities, Sun Yat-sen University (24lgqb019).

\bibliographystyle{ieeetr}

\end{document}